\documentclass[runningheads]{llncs}
\usepackage{url}
\usepackage{breakurl}
\usepackage[breaklinks]{hyperref}
\usepackage[misc]{ifsym}
\usepackage{graphicx}
\usepackage{amssymb}
\usepackage{hyperref}
\usepackage[mathscr]{euscript}
\usepackage{amsmath}
\usepackage{algorithm}
\usepackage[noend]{algpseudocode}
\usepackage{amssymb,amsmath}
\usepackage[capitalize,nameinlink]{cleveref}
\usepackage{lipsum}
\begin{document}

\title{Reoptimization of Path Vertex Cover Problem}

\author{Mehul Kumar \and Amit Kumar\textsuperscript{(\Letter)} \and C. Pandu Rangan}

\authorrunning{M. Kumar et al.}

\institute{Department of Computer Science and Engineering,\\
Indian Institute of Technology Madras, Chennai, India\\
\email{\{mehul,amitkr,rangan\}@cse.iitm.ac.in}}

\maketitle             

\begin{abstract}
Most optimization problems are notoriously hard. Considerable efforts must be spent in obtaining an optimal solution to certain instances that we encounter in the real world scenarios. Often it turns out that input instances get modified locally in some small ways due to changes in the application world. The natural question here is, given an optimal solution for an old instance $I_O$, can we construct an optimal solution for the new instance $I_N$, where $I_N$ is the instance $I_O$ with some local modifications. Reoptimization of NP-hard optimization problem precisely addresses this concern. It turns out that for some reoptimization versions of the NP-hard problems, we may only hope to obtain an approximate solution to a new instance. In this paper, we specifically address the reoptimization of path vertex cover problem. The objective in $k$-$path$ vertex cover problem is to compute a minimum subset $S$ of the vertices in a graph $G$ such that after removal of $S$ from $G$ there is no path with $k$ vertices in the graph. We show that when a constant number of vertices are inserted, reoptimizing unweighted $k$-$path$ vertex cover problem admits a PTAS. For weighted $3$-$path$ vertex cover problem, we show that when a constant number of vertices are inserted, the reoptimization algorithm achieves an approximation factor of $1.5$, hence an improvement from known $2$-approximation algorithm for the optimization version. We provide reoptimization algorithm for weighted $k$-$path$ vertex cover problem $(k \geq 4)$ on bounded degree graphs, which is also an NP-hard problem. Given a $\rho$-approximation algorithm for $k$-$path$ vertex cover problem on bounded degree graphs, we show that it can be reoptimized within an approximation factor of $(2-\frac{1}{\rho})$ under constant number of vertex insertions.
\end{abstract}

\section*{{\large{1.1}} \ Introduction}
Most combinatorial optimization problems are NP-hard. Efficient algorithms to find an optimal solution for such problems are not known. By efficient, we mean running in time polynomial in the input size. Hence, we resort to approximation algorithms which aim to efficiently provide a near-optimal solution. For minimization problems, a $\rho$-approximation algorithm ($\rho >1$) efficiently outputs a solution of cost at most $\rho$ times the optimum, where $\rho$ is called the approximation ratio. A family of $(1+\epsilon)$ approximation algorithms ($\forall \epsilon >0)$ with polynomial running times is called a polynomial time approximation scheme (PTAS).

In many practical applications, the problem instance can arise from small perturbations in the previous instance of an optimization problem. A naive approach is to work on the new problem instance from scratch using known $\rho$-approximation algorithm. But, with some prior knowledge of the solution for old instance, can we perform better? The computational paradigm of reoptimization addresses this question.

We consider the case where one has devoted a substantial amount of time to obtain an exact solution for the NP-hard optimization problem. Now, the goal is to reoptimize the solution whenever the modified instance is known. A reoptimization problem $Reopt(\pi)$ can be built over any optimization problem $\pi$. An input instance for $Reopt(\pi)$ is a triple $(I_N,I_O,OPT(I_O))$, where $I_O$ is an old instance, $I_N$ is a modified instance and $OPT(I_O)$ is an optimal solution for $\pi$ on $I_O$.

Suppose $I_N$ is a hard instance obtained via some perturbations in $I_O$ and assume that we have an optimal solution of $I_O$. The natural question is, can we find an optimal solution of $I_N$? In general this may not be the case and we specifically show in Lemma 3 that, if the optimization problem is $path \; vertex \; cover$ and the perturbation is a single vertex insertion, then even possessing $OPT(I_O)$ does not help to find an optimal solution for $I_N$ efficiently, unless $P=NP$. Hence, the objective of an efficient algorithm for $Reopt(\pi)$ is to either achieve a better approximation ratio or improve the running time of the known approximation algorithm. In this paper, the optimization problem we consider for reoptimization is the path vertex cover problem. This problem has its applications in traffic control and secure communication in wireless networks \cite{DBLP:conf/wistp/Novotny10}. We briefly explain the optimization problem below:

A path of order $k$ in a graph is a simple path containing $k$ vertices. For a given graph $G=(V,E)$, $S \subseteq V$ is a feasible $k$-$path$ vertex cover iff every path of order k in $G$ contains at least one vertex from $S$. The problem of finding a feasible k-path vertex cover on a graph is known as $k$-path vertex cover problem ($k$-$PVCP$). This problem has two variants: weighted and unweighted. The goal in unweighted $k$-$PVCP$ is to find a feasible subset of minimum cardinality whereas in weighted $k$-$PVCP$, the objective is to find minimum weighted subset of vertices that covers all the paths of order $k$ or more.

\section*{{\large{1.2}} \ Related Work and Contributions}
For any fixed integer $k \geq 2$, the $k$-$path$ vertex cover problem ($k$-$PVCP$) is known to be NP-complete for an arbitrary graph $G$ and also it's NP-hard to approximate it within a factor of $1.3606$, unless P=NP \cite{DBLP:journals/dam/BresarKKS11}. However, unweighted and weighted $k$-$path$ vertex cover problems on trees have polynomial time algorithms \cite{DBLP:journals/dam/BresarKKS11} \cite{DBLP:journals/dam/BresarKSS14}. The problem has been studied in \cite{DBLP:conf/soda/Lee17} as $k$-$path$ traversal problem which presents a $\log(k)$-approximation algorithm for the unweighted version. For $k=2$, the $k$-$PVCP$ corresponds to the conventional vertex cover problem. The $3$-$PVCP$ is a dual problem to the dissociation number of the graph. Dissociation number is the maximum cardinality of a subset of vertices that induce a subgraph with maximum degree at most 1. \cite{DBLP:journals/tcs/TuZ11} provides a $2$-approximation algorithm for weighted $3$-$PVCP$ and there is a $3$-approximation algorithm for $4$-$PVCP$ \cite{camby2014primal}. 

For the reoptimization version, G. Ausiello et al. present an algorithm for reoptimizing unweighted vertex cover problem. Following the approach in \cite{article}, section 2 shows that reoptimization of unweighted $k$-$PVCP$ admits a PTAS under the constant number of vertex insertions. In section 3.1, we extend the reoptimization paradigm for weighted vertex cover problem in \cite{article} to weighted $k$-$PVCP$. As a use case for the subroutine in section 3.1, we show in section 3.2 that weighted $3$-$PVCP$ can be reoptimized with an approximation factor of 1.5 under constant number of vertex insertions. In section 3.3, we present an algorithm for reoptimization version of weighted $k$-$PVCP$ $(k \geq 4)$ on bounded degree graphs under constant number of vertex insertions. For a given $\rho$-approximation algorithm for weighted $k$-$PVCP$ $(k \geq 4)$, this algorithm achieves an approximation ratio of $(2-\frac{1}{\rho})$  for such graphs. In Appendix A, we present the hardness results. In Appendix B, we present $n$ and $k$-approximation algorithms that are used in our algorithm for weighted $k$-$PVCP$ $(k \geq 5)$.

\section*{{\large{1.3}} \ Preliminaries}
In this paper, the graphs we consider are simple undirected graphs. A graph $G$ is a pair of sets $(V, E)$, where $V$ is the set of vertices and $E$ is the set of edges formed by unordered pairs of distinct vertices in $V$. For a vertex $v \in V$, we denote the set of neighbours of $v$ in $G$ by $N_G(v)$, where $N_G(v) = \{u \in V \mid (u,v) \in E\}$. For any $S \subseteq V$, we define $N_G(S)$ to be the neighbouring set of S in $G$, where $N_G(S) \subseteq (V-S)$ and $\forall u \in N_G(S) \; \exists v \in S$ such that $(u,v) \in E$. For any $S \subseteq V$, we use $G[S]$ to represent the subgraph induced on the vertex set $S$ in $G$. Let $V(G)$ and $E(G)$ denote the vertex set and edge set of $G$ respectively. A degree of a vertex is the number of edges incident on it. We use $\Delta(G)$ to denote the maximum degree of the vertices in graph $G$. In the case of weighted graphs, with every vertex we associate a positive weight function $f:V\rightarrow \mathbb{R}^+$. For any $v\in V$, let $w(v)$ be the weight of the vertex and for any subset $S\subseteq V$, the weight of the subset $w(S)$ is $\sum_{v\in S} w(v)$. Size of a graph is defined as the number of vertices in it. A constant-size graph is a graph where number of vertices are constant and independent of input parameters of the algorithm. Two graphs are said to be disjoint if they do not share any common vertices.

Let $G=(V,E)$ and $G_A = (V_A,E_A)$ be two graphs where $V\cap V_A = \phi$. Given a set of attachment edges $E^a \subseteq (V \times V_A)$, insertion of $G_A$ into $G$ yields the undirected graph $G' = (V', E')$, where $V' = V \cup V_A$ and $E' = E \cup E_A \cup E^a$. Thus, a constant number of vertex insertions can be realized as a constant-size graph insertion. We define a vertex insertion in $G$ as a special case of graph insertion where the inserted graph $G_A$ is a single vertex $v \notin V[G]$. In general, we denote $OPT(G)$ as the optimal solution and $ALG(G)$ as the solution output by an algorithm for the corresponding problem on $G$.

Let $\pi$ denote the optimization problem and $Reopt(\pi)$ is the reoptimization version of it. The $\pi$ we consider in this paper is the $k$-$path$ vertex cover problem. The input instance of $Reopt(\pi)$ is a triple $(G_O, G_N, OPT(G_O))$, where $G_O$ is the old graph, $G_N$ is the new graph and $OPT(G_O)$ is an optimal solution for $\pi$ on $G_O$. Let $A_\rho(\pi)$ be a known $\rho$-approximation algorithm  for $\pi$.
For the algorithms we give, the equality statements are considered as assignment from right to left.

\section*{{\large{2}} \ Reoptimization of unweighted $k$-$PVCP$}
Let $\pi$ be unweighted $k$-$PVCP$. We consider the reoptimization version $Reopt(\pi)$ where a constant-size graph $G_A=(V_A,E_A)$ is inserted to the old graph $G_O = (V_O, E_O)$ to yield the new graph $G_N = (V_N, E_N)$. Let $|V_A|=c$. For a given $\epsilon$, we design an algorithm $Unwtd$-$kpath$ for $Reopt(\pi)$ that outputs $ALG(G_N)$ as a solution.

\begin{algorithm}
\caption{$Unwtd$-$kpath(G_O, G_N, OPT(G_O), \epsilon)$}\label{alg:Unwtd$-$kpath}
\begin{algorithmic}[1]
\State $V_A = V(G_N) - V(G_O)$
\State $c = |V_A|$
\State $m = \lceil c/\epsilon \rceil$
\State \texttt{$S_1 = V(G_N)$}
\For{\texttt{each subset $X$ of $V(G_N)$ where $|X| \leq m$}}
\State \textbf{if}\texttt{($X$ covers all $k$-$paths$ in $G_N$ and $|X| < |S_1|$)}
\State \texttt{$\;\;\;\;S_1 = X$}
\EndFor
\State \texttt{$S_2 = OPT(G_O) \cup V_A$}
\State \texttt{$ALG(G_N) = min(|S_1|, |S_2|)$}

\State \textbf{return} $ALG(G_N)$
\end{algorithmic}
\end{algorithm}

\begin{theorem}
$Unwtd$-$kpath$ for $Reopt(\pi)$ under constant-size graph insertion admits a PTAS.
\end{theorem}

\begin{proof}
Since $OPT(G_N)\cap V(G_O)$ and $OPT(G_O) \cup V_A$ is a feasible $k$-$path$ vertex cover on $G_O$ and $G_N$ respectively, we get
$$ |OPT(G_O)| \; \leq |OPT(G_N)| \; \leq |OPT(G_O)| + c \; \; \; \; \cdots \; (1) $$
If $OPT(G_N)$ has size at most $m$, it would have been found in step 7 of $Unwtd$-$kpath$. We know,
$$ |ALG(G_N)| \; \leq \; |OPT(G_O)| + c\; = |S_2| $$
and $S_2$ is picked when $|OPT(G_N)| \geq m \geq \frac{c}{\epsilon}$. Thus, approximation factor for $ALG(G_N)$ using inequality (1) and above observation is,
$$ \frac{|ALG(G_N)|}{|OPT(G_N)|} \; \leq \; \frac{|OPT(G_O)| + c}{|OPT(G_N)|} \; \leq \; \frac{|OPT(G_N)| + c}{|OPT(G_N)|}  \; \leq \; 1 \; + \; \epsilon $$

\end{proof}
Further, we analyze the runtime. Enumerating all possible $k$-$paths$ in a graph of $n$ vertices takes $O(n^k)$ time. Thus for a given set $X$, we can decide in polynomial time whether all paths of order $k$ are covered by the set. There are $O(n^m)$ subsets of size at most $m$, where $n = |V_N|$. The runtime of the algorithm is $O(n^m \cdot n^k) = O(n^{\frac{c}{\epsilon}} \cdot n^k)$, and hence a valid PTAS. Note that the runtime can be improved by using color coding algorithm for finding a $k$-$path$ \cite{DBLP:conf/stoc/AlonYZ94}, which runs in $O(2^k n^{O(1)})$ time.

\section*{{\large{3.1}} \ Subroutine for Reoptimzation of Weighted $k$-$PVCP$}
Let $\pi_k$ be weighted $k$-$PVCP$. $A_\rho(\pi_k)$ be a known $\rho$-approximation algorithm  for $\pi_k$. In reoptimization version of the problem $Reopt(\pi_k)$, a new graph $G_N$ is obtained by inserting a graph $G_A$ to $G_O$. \\\\
\textbf{Definition}: A family $\mathscr{F} = \{F_1, F_2, \cdots ,F_\psi\}$, where $\psi = |\mathscr{F}|$ and of subsets of $V_N$ is called a \textbf{good} family if it satisfies the following two properties:
\let\labelitemi\labelitemii
\begin{itemize}
    \item \textbf{Property 1:} $\exists$ $F_i \in \mathscr{F}$ such that  $F_i \subseteq OPT(G_N)$ and,
    \item \textbf{Property 2:} $\forall$ $F_i  \in \mathscr{F}$, $F_i$ covers all the $k$-$paths$ which contains at-least one vertex from $V(G_A)$ in graph $G_N$.
\end{itemize}

We give below a generic algorithm that works on the good family $\mathscr{F}$. This family of sets will be constructed in different ways for different problems. The details are provided in the respective sections. 

An algorithm for $Reopt(\pi_k)$ constructs the good family $\mathscr{F}$ and feeds it to the subroutine $Construct$-$Sol$. The algorithm $Construct$-$Sol$ iteratively prepares a solution $S_i$ for each set $F_i \in \mathscr{F}$. The inputs to the algorithm $Construct$-$Sol$ are: modified graph $G_N$, inserted graph $G_A$, old optimal solution $OPT(G_O))$, a good family $\mathscr{F}$ and $A_\rho(\pi_k)$.
\begin{algorithm}
\caption{$Construct$-$Sol(G_N,G_A,OPT(G_O),\mathscr{F},A_\rho(\pi_k))$}\label{alg:Construct$-$Sol}
\begin{algorithmic}[1]
\For{\texttt{$i = 1$ to $|\mathscr{F}|$}}
        \State \texttt{$S^1_i = OPT(G_O) \cup F_i$}
        \State \texttt{$G' = G_N[(V_N - V(G_A)) - F_i]$}
        \State \texttt{Run $A_\rho(\pi_k)$ on $G'$ and denote the output set as $S^2_i$}
        \State \texttt{$S^2_i = S^2_i \cup F_i$}
        \State \texttt{$S_i = minWeight(S^1_i , S^2_i)$}
\EndFor
\State $ALG(G_N) = minWeight(S_1, S_2, \dots, S_{|\mathscr{F}|})$
\State \textbf{return} $ALG(G_N)$
\end{algorithmic}
\end{algorithm}

\begin{lemma} \label{Le_1}
If $OPT(G)$ is an optimal solution for weighted $k$-$PVCP$ for $G$, then for any $S \subseteq OPT(G)$, $w(OPT(G[V - S])) \leq w(OPT(G)) - w(S)$.
\end{lemma}
\begin{proof}
If $F$ is a feasible $k$-$path$ cover for $G[V]$, then for any $V^* \subseteq V$, $F \cap V^*$ is a feasible $k$-$path$ cover for $G[V^*]$.

$OPT(G) - S$ is a feasible solution for $G[V - S]$ because $(V-S) \cap OPT(G) = OPT(G) - S$. Since $S \subseteq OPT(G)$, $w(OPT(G) - S) \; = \; w(OPT(G)) - w(S)$. Hence, $w(OPT(G[V - S])) \leq w(OPT(G)) - w(S)$.
\end{proof}

\begin{theorem}\label{thm_2}
The algorithm $Construct$-$Sol$ outputs a solution $ALG(G_N)$ with an approximation factor of $(2-\frac{1}{\rho})$, running in $O({|V(G_N)|}^2 \cdot \psi \cdot T(A_\rho(\pi_k),G_N))$ steps, where $\rho$ is the approximation factor of a known $A_\rho(\pi_k)$.
\end{theorem}
\begin{proof}
A graph $G_A$ is inserted to $G_O$ to yield the new graph $G_N$. By property 1 of the good family $\mathscr{F}$, the optimal solution for $G_N$ must include at least one set in $\mathscr{F}=\{F_1, \dots , F_{\psi}\}$, where $\psi=|\mathscr{F}|$. Thus, at least one $S_i (1\leq i \leq \psi)$ is prepared by the subroutine.

Let $OPT(G_N)_i$ be the optimal solution which includes $F_i$ and not $(V(G_A) - F_i)$. We prepare $\psi$ number of solutions for the graph $G_N$.

$S^1_i$ is a feasible $k$-$path$ cover for $G_N$, where feasibility follows from property 2 of the family. We can write the following inequalities:
$$ w(OPT(G_O)) \leq w(OPT(G_N)_i)$$
$$ w(S^1_i) = w(OPT(G_O) \cup F_i) \leq w(OPT(G_O)) + w(F_i)$$
 From above two inequalities,
 $$w(S^1_i)  \leq w(OPT(G_N)_i) + w(F_i)  \; \; \; \; \cdots (1)$$
 
 Another solution $S^2_i$ is prepared. From Lemma \ref{Le_1} and construction of $S^2_i$, we can write the following inequality:
$$w(S^2_i) \leq \rho(w(OPT(G_N)_i) - w(F_i)) + w(F_i) \; \; \; \; \cdots (2)$$
Since $\rho > 1$, adding $(\rho-1) \times (1)$ and $(2)$, we get
$$(\rho -1)w(S^1_i) + w(S^2_i) \leq (2\rho -1 )(w(OPT(G_N)_i))$$
Minimum weighted subset between $S^1_i$ and $S^2_i$ is chosen to be $S_i$. Then,
$$(\rho -1)w(S_i) + w(S_i) \leq (2\rho -1 )(w(OPT(G_N)_i))$$
$$\implies \forall i \in [1,\psi],  w(S_i) \leq (2 -\frac{1}{\rho} )(w(OPT(G_N)_i))\; \; \; \; \cdots (3)$$

We have prepared a set of $\psi$ number of solutions that is, $\{S_1, S_2, \cdots ,S_\psi\}$. By definition of $OPT(G_N)_i$ and property 2 of good family $\mathscr{F}$, we get that, if $F_i \subseteq OPT(G_N)$, then $w(OPT(G_N)_i) = w(OPT(G_N))$. By the property 1 of $\mathscr{F}$, there exists an $F_i$ such that $F_i \subseteq OPT(G_N)_i$. Hence, the following inequality for such an $i$ holds true:
$$ w(S_i) \leq (2 -\frac{1}{\rho} )(w(OPT(G_N)_i)) = (2 -\frac{1}{\rho} )(w(OPT(G_N)))$$
We know, $\forall i \in [1,\psi], \; w(ALG(G_N)) \leq w(S_i)$. So, 
$$ w(ALG(G_N)) \leq (2-\frac{1}{\rho}) (w(OPT(G_N)))$$
Thus, algorithm $Construct$-$Sol$ outputs a solution with an approximation factor of $(2-\frac{1}{\rho})$. Note that step 3 of the algorithm takes $O(|V_N|^2)$ time. Moreover, if the running time of $A_\rho(\pi_k)$ on input graph $G_N$ is $T(A_\rho(\pi_k),G_N)$, then the running time of algorithm $Construct$-$Sol$ is $O({|V(G_N)|}^2 \cdot \psi \cdot T(A_\rho(\pi_k),G_N))$.
\end{proof}

\section*{{\large{3.2}} \ Reoptimization of weighted $3$-$PVCP$}
Let $\pi_3$ be weighted $3$-$PVCP$. A constant-size graph $G_A=(V_A,E_A)$ is inserted to $G_O$ to yield the new graph $G_N = (V_N,E_N)$. Let $T(A_2(\pi_3),G_N)$ denote the runtime of $2$-approximation algorithm \cite{DBLP:journals/tcs/TuZ11} for $\pi_3$. Let $|V_N| = n$.

\begin{algorithm}
\caption{$Wtd$-$3pathGI(G_N,G_A,OPT(G_O),A_2(\pi_3))$}\label{alg:Wtd$-$3path}
\begin{algorithmic}[1]

\State $\mathscr{F} = \phi$ 
\For{\textbf{all} \texttt{$X \subseteq (V_A = V(G_A))$ and $X$ is a 3-path cover for $G_A$}}
    \State \texttt{$V_I$ = Set of isolated $v$ in $G_A[V_A-X]$ and $N_{G_N}(v) \cap V(G_O) \neq \phi$}   
    \State \texttt{$E_I$ = Set of isolated $(u,v)$ in $G_A[V_A-X]$ and $N_{G_N}(\{u,v\}) \cap V(G_O) \neq \phi$}
    \For{\textbf{all} \texttt{ $ (u,v) \in E_I$}}
        \State \texttt{$X = X \cup N_{G_O}(\{u,v\})$}
    \EndFor
    \State \texttt{$Y = N_{G_N}(V_A) - X$}
    \For{\textbf{all} \texttt{$Y' \subseteq Y$, $|Y'|\leq |V_I|$  }}
    \State \textbf{if}\texttt{($Y-Y'$) is a $3$-$path$ cover for $G_N[V_I \cup Y]$}
        \State $\;\;\;\;X'= X \cup (Y-Y')$
            \State $\;\;\;\;X' = X' \cup N_{G_O}(Y')$
            \State \texttt{$\;\;\;\;\mathscr{F} = \mathscr{F} \cup \{X'\}$}
    \EndFor
\EndFor
\State \textbf{return} $Construct$-$Sol(G_N,G_A,OPT(G_O),\mathscr{F},A_2(\pi_3))$
\end{algorithmic}
\end{algorithm}

\begin{theorem}
Algorithm $Wtd$-$3path$ is a $1.5$ approximation for $Reopt(\pi_3)$ under constant-size graph insertion.
\end{theorem}
\begin{proof}

The algorithm works in 3 phases to construct the good family $\mathscr{F}$. The algorithm prepares a subset $X$ for each feasible $3$-$path$ cover $X$ for $G_A$ because the optimal solution for $G_N$ must contain one subset among the feasible $X$'s. In the first phase, if an edge in $G_A[V_A - X]$ has a neighbour in $G_O$, it must be included in $X$.

Let $Y = N_{G_N}(V_A) - X$. In the second phase , $G_N[V_I \cup Y]$ is made free from all $3$-$paths$ by removing a feasible subset $Y-Y'$. In the third phase, the neighbours of the vertices in $Y'$ are included in $X'$ because such a neighbour will form $3$-$path$ with the vertices in $Y'$ and $V_I$.

Since we consider all feasible subsets $X$ and feasible $Y-Y'$ for the corresponding $X$, the constructed family $\mathscr{F}$ satisfies both the properties of good family. Thus from Theorem 2 and $A_2(\pi_3)$, we get the desired approximation. Further, we analyze the running time. Let $c = |V(G_A)|$. The maximum cardinality of $Y$ is $n$. Then, the steps 3 to 12 in the algorithm run in $O(n^{3} n^{c+1})$ because $|Y'| \leq c$. Thus $|\mathscr{F}| \in O(c^2 \cdot 2^c \cdot n^{c+4})$. Hence the algorithm $Wtd$-$3path$ runs in $O(n^{c+6} \cdot 2^{c} \cdot T(A_2(\pi_3),G))$.
\end{proof}

\section*{{\large{3.3}} \ Reoptimization of weighted $k$-$PVCP$ $(k\geq 4)$ for bounded degree graphs}
A graph free from $2$-$paths$ contains only isolated vertices. A graph that does not have any $3$-$path$ contains isolated vertices and isolated edges. But, in the case of graphs that are free from $k$-$paths$ ($k\geq4$), star graph is a possible component. As the number of subsets needed to be considered for preparation of $\mathscr{F}$ would be exponential in the vertex degree, we restrict the reoptimization of weighted $k$-$PVCP$ $(k \geq 4)$ to bounded degree graphs. Lemma 3 in Appendix A shows that the problem on bounded degree graphs is NP-complete. 

The local modification which we consider for reoptimization is constant-size graph insertion. Let $G_O = (V_O, E_O)$ be the old graph. Given $G_O$, constant-size graph $G_A = (V_A,E_A)$ and attachment edges $E^a$, the new graph $G_N = (V_N,E_N)$ is obtained. Let $|V_N| = n$ and $|V_A| = c$. Let the maximum degree of the graph $G_N$ be $\Delta$. We use $P_k(G,V')$ to denote the collection of $k$-$paths$ in graph $G$ containing at least one vertex from $V' \cap V(G)$. For a set of vertices $V'$, a graph is said to be $V'$-connected graph if every connected component in the graph contains at least one vertex from $V'$. Let $A_\rho(\pi_k)$ be a $\rho$-approximation algorithm for weighted $k$-$PVCP$ $(\pi_k)$ and $T(A_{\rho}(\pi_k),G_N)$ be the running time of $A_\rho(\pi_k)$ on $G_N$. \\\\
\textbf{Definition:} We define a variation of $BFS$ on a graph $G$, where traversal starts by enqueuing a set of vertices $V'$ instead of a single root vertex. Initially, all the vertices in $V'$ are at the same level and unvisited. Now the unvisited nodes are explored in breadth first manner. In this variation, we obtain the $BFS$ forest for the input $(G,V')$, where the vertices of $V'$ are at level 1 and the subsequent levels signify the order in which the vertices are explored.

Consider the $BFS$ forest obtained from $V_A$ in $G_N$. We use $L_i$ to denote the set of vertices at level $i$ $(i \geq 0)$ of the $BFS$ forest. Let $S_j = \bigcup _{i = 0}^j L_i$, where $L_0 = \phi$ and $L_1 = V_A$. Then $L_i = N_{G_N}(S_{i-1})$ for $i \geq 2$. Note that this $BFS$ forest has $|V_A|$ number of disjoint $BFS$ trees, where the trees have distinct root vertices from $V_A$.

\begin{lemma} \label{Le_4}
In a $BFS$ forest obtained after performing $BFS$ traversal from a set of vertices $V_A \subseteq V$ in a graph $G=(V,E)$ having no $k$-$paths$, the number of vertices at each level is at most $|V_A|\Delta (\Delta-1)^{\lceil \frac{k-5}{2} \rceil}$.
\end{lemma}
\begin{proof}
Consider the case when $BFS$ is performed from a single vertex set $V_A = \{v_1\}$ to obtain a $BFS$ tree.
For any level $i$, $|L_i| \leq \Delta (\Delta-1)^{i-2}$. Thus the statement holds true for $i\leq \lceil \frac{k-1}{2} \rceil$. 
For the case when $i > \lceil \frac{k-1}{2} \rceil$, let $j = i - \lceil \frac{k-3}{2} \rceil$. We claim that there exists a vertex $v$ in $L_j$ such that $v$ is a common ancestor for all the vertices in $L_i$. Assume to contrary that the claim is false. If $|L_i|=1$ the claim is trivially true. Otherwise we have two distinct vertices $v_x$ and $v_y \in L_i$ such that they have the lowest common ancestor in $L_{j'}$, where $1 \leq j' \leq j-1 = i - \lceil \frac{k-1}{2} \rceil$. This imposes a path $\langle v_i,\cdots v, \cdots v_j \rangle$ of order $\lceil \frac{k-1}{2} \rceil + 1 + \lceil \frac{k-1}{2} \rceil \geq k$. But it contradicts the fact that $G$ has no paths of order $k$ or more. Hence $|L_i| \leq \Delta (\Delta-1)^{i-j-1} = \Delta (\Delta-1)^{\lceil \frac{k-5}{2} \rceil}$.

Now, when $BFS$ is performed for the case when $|V_A| > 1$, the $BFS$ forest obtained has $|V_A|$ number of disjoint $BFS$ trees where each tree satisfies the above argument. Hence the number vertices in each level in the $BFS$ forest is at most $|V_A| \Delta (\Delta-1)^{\lceil \frac{k-5}{2} \rceil}$
\end{proof}

\begin{algorithm}
\caption{$Construct$-$F(X,V,L,level,V_A,G_N,\mathscr{F},k)$}\label{alg:Wtd$-$kpath}
\begin{algorithmic}[1]
\State \texttt{$\mathscr{F} = \mathscr{F} \cup \{X \cup L\}$}
\State \textbf{if} $level \geq k-1$
\State \textbf{$\;\;\;\;\textbf{return} \  \mathscr{F}$}

\State $b = (|V_A| \Delta (\Delta-1)^{\lceil \frac{k-5}{2} \rceil})$
\For{\texttt{each non-empty subset $V'$ of $L$ and $|V'| \leq b$}}
        \State \textbf{if}\texttt{ $G_N[V \cup V']$ is a k-path free $V_A$-connected graph}
        \State \texttt{     $\;\;\;X'' = X \cup (L-V')$}
        \State \texttt{     $\;\;\;V'' = V \cup V'$}
        \State \texttt{     $\;\;\;L'' = N_{G_N}(V'') - X''$}
        \State \texttt{     $\;\;\;\mathscr{F} = Construct$-$(X'', V'',L'',level+1,V_A,G_N,\mathscr{F},k)$}
\EndFor
\State \textbf{$\textbf{return} \  \mathscr{F}$}
\end{algorithmic}
\end{algorithm}

\begin{algorithm}
\caption{$Wtd$-$kpath(G_N,G_A,OPT(G_O),A_\rho(\pi_k),k)$}\label{alg:Wtd$-$kpath}
\begin{algorithmic}[1]
\State Initialization: $\mathscr{F} = \phi$, $level =1$, $X=\phi$ and $V=\phi$.
\State $\mathscr{F} = Construct$-$F(X,V,V(G_A),level,V(G_A),G_N,\mathscr{F},k)$

\State $ALG(G_N) = Construct$-$Sol(G_N,G_A,OPT(G_O),\mathscr{F},A_\rho(\pi_k))$
\State \textbf{return}
$ALG(G_N)$
\end{algorithmic}
\end{algorithm}

\begin{theorem}
Algorithm $Wtd$-$kpath$ is a $(2-\frac{1}{\rho})$ approximation for $Reopt(\pi_k)$ under graph insertion and runs in $O(n^{O(1)} \cdot {2}^{k (\Delta+1) b }\cdot T(A_{\rho}(\pi_k),G_N))$, where $b=|V_A|\Delta (\Delta-1)^{\lceil \frac{k-5}{2} \rceil}$ and $ ((\Delta+1) b) \in O(\log n)$.
\end{theorem}

\begin{proof}
We first prove that for every call to the function \\ $Construct$-$F(X,V,L,level,V_A,G_N,\mathscr{F},k)$, the following invariants on $V$, $X$ and $L$ are maintained:
\let\labelitemi\labelitemii
\begin{itemize}
    \item $V \subseteq S_{level}$ and $G_N[V]$ is a k-path free $V_A$-connected graph.
    \item $X$ is the set of neighbours of $V$ in $G_N[S_{level}]$
    \item L is the set of neighbours of $V$ in graph $G_N$ that are also in $L_{level+1}$, i.e $L = N_{G_N}(V)-S_{level}$.
\end{itemize}
The above invariants trivially hold true during the first call to the function $Construct$-$F$. Assuming the invariants to be true during a call to $Construct$-$F$, we show that the subsequent recursive calls maintain the invariants. Note that the parameter $'level'$ is incremented to $level+1$ during the recursive call. $G_N[V \cup V'] = G_N[V'']$ is a k-path free $V_A$-connected subgraph in $G_N$. Also, $V'' \subseteq S_{level+1}$ because $V \subseteq S_{level}$ and $L \subseteq L_{level+1}$. The invariance property of $X$ and $L$ implies $X'' = X \cup (L-V')$ is the set of neighbours of $V''$ in $G_N[S_{level+1}]$. From previous observation about $X''$ and $V''$, we get that $L''=N_{G_N}(V'') - X''$ is the set of neighbours of $V''$ in $G_N$  which are also in $L_{level+2}$. Thus, the invariants are maintained. 

Note that $X$ covers all the paths in $P_k(G_N[S_{level}], V_A)$. $X \cup L$ is a $k$-$path$ cover for $G_N$ because the paths in $P_k(G_N , V_A) - P_k(G_N[S_{level}], V_A)$ contain at least one vertex from $L$. Thus, $\{X \cup L\}$ is included in $\mathscr{F}$ to satisfy property 2 of good family.

By Lemma 2, it is sufficient to consider non empty subsets $V'$ of size at most $(|V_A|\Delta (\Delta-1)^{\lceil \frac{k-5}{2} \rceil})$ from subsequent level $L$ to construct $V''$.  For each recursive call, the case when $V'$ or $L$ is empty is handled in the step 1. A $V_A$-connected graph that has no $k$-$paths$ will have a maximum level of $k-1$ in the $BFS$ forest. The algorithm explores all feasible subsets $V'$ for each $level \leq k-1$. Thus the property 1 of good family holds true for $\mathscr{F}$, because the family includes all possibilities for $\{X\cup L\}$ that covers $P_k(G_N, V_A)$. Thus, the constructed family $\mathscr{F}$ is indeed a good family.

Let $RT(l)$ be the running time of the function $Construct$-$F$, where $l$ is the parameter $'level'$. Let $C = \Sigma_{i=1}^{i=b} {{\Delta b} \choose {i}}$. Observe that $|L| \leq (\Delta \cdot b)$ due to the construction of $L''$ in the previous recursion. As we are choosing sets of size at most $b$ from $L$, we get the recursion $RT(l) = O(n^{O(1)} \cdot C^k \cdot RT(l+1))$ for $1 \leq l \leq (k-1)$ and $RT(k)=O(n^{O(1)})$. Thus step 6 in $Wtd$-$kpath$ runs in $O(n^{O(1)} \cdot C^k)$ time. In each function call, $|\mathscr{F}|$ is incremented by one element. Thus, $|\mathscr{F}| \leq 2 ^{kb}$ because $|V''|\leq b$ for each level. Note that $C \leq 2^{b\Delta}$. Hence using Theorem $2$, the algorithm $Weighted$-$kpath$ runs in $O(n^{O(1)} \cdot C^k \cdot {2} ^{k b} \cdot T(A_{\rho}(\pi_k),G_N))$ = $O(n^{O(1)} \cdot {2} ^{k(\Delta+1) b} \cdot T(A_{\rho}(\pi_k),G_N))$ and achieves the desired approximation.

\end{proof}

Using $3$-$approximation$ algorithm for weighted $4$-$PVCP$  \cite{camby2014primal} and Theorem $4$, we get the following corollary:
\begin{corollary}
Algorithm $Wtd$-$kpath$ is a $\frac{5}{3}$ approximation for $Reopt(\pi_4)$ under constant-size graph insertion, where $(\Delta) \in O(1)$  and $A_\rho(\pi_4)$ is $A_3(\pi_4)$.
\end{corollary}
In Appendix B, we present approximation algorithms for weighted $k$-$path$ vertex cover problem. The $n$-approximation algorithm runs in $O(2^k n^{O(1)})$ time, where $n$ is the size of input graph. The $k$-approximation algorithm is a primal dual based algorithm. Using Theorems 5 and 6 of Appendix, we get the following corollaries:

\begin{corollary}
Algorithm $Wtd$-$kpath$ is a $(2 - \frac{1}{n})$ approximation for $Reopt(\pi_k)$ under constant-size graph insertion, where $(\Delta) \in O(1)$  and $A_\rho(\pi_k)$ is $A_n(\pi_k)$.
\end{corollary}

\begin{corollary}
Algorithm $Wtd$-$kpath$ is a $(2 - \frac{1}{k})$ approximation for $Reopt(\pi_k)$ under constant-size graph insertion, where $(\Delta) \in O(1)$ and $A_\rho(\pi_k)$ is $A_k(\pi_k)$.
\end{corollary}

Note that the algorithm only explores the vertices till level $k-1$ that is, the vertices in the set $S_{k-1}$. Thus $|\mathscr{F}|$ is at most $2^{|S_{k-1}|}$. Therefore the algorithm will also run efficiently for the scenarios where the graph $G_A$ is attached to a $'sparse'$ part of $G_O$, that is for $|S_{k-1}| \in O(\log n)$.

\begin{corollary}
Algorithm $Wtd$-$kpath$ is a $(2 - \frac{1}{\rho})$ approximation for $Reopt(\pi_k)$ under graph insertion, where $|S_{k-1}| \in O(\log n)$.
\end{corollary}

\section*{{\large{4}} \ Concluding Remarks}
In this paper, we have given a PTAS for reoptimization of unweighted $k$-$PVCP$ under constant number of vertex insertions. When constant-size graph is inserted to the old graph, we have presented $1.5$-approximation algorithm for reoptimization of weighted $3$-$PVCP$. Restricting our inputs to bounded degree graphs, we have presented $\frac{5}{3}$- approximation for reoptimization of weighted $4$-$PVCP$ under constant-size graph insertion. For the reasons we mentioned in Section 3.3, our technique for reoptimization of weighted $k$-$PVCP$ $(k\geq 4)$ cannot be extended to arbitrary graphs. Hence, reoptimization of weighted $k$-$PVCP$ $(k\geq 4)$ for arbitrary graphs under constant number of vertex insertions is an intriguing open problem.\\\\\\
\noindent{\textbf{\large{Acknowledgment}}}\\\\
We thank Narayanaswamy N S for enlightening discussions on the problem.
\newpage

\bibliographystyle{splncs04}
\bibliography{ReoptPVC}
\newpage
\appendix
\noindent{\textbf{\large{Appendix}}}
\section{Reductions}
\begin{lemma}\label{lem_2}
Minimum unweighted $k$-$path$ vertex cover problem on graphs with maximum degree $\Delta(G) \geq 3$ is NP-complete.
\end{lemma}
\begin{proof}
$k$-$path$ vertex cover is in NP as enumerating over all paths of order $k$ would verify a $k$-$path$ vertex cover instance, where run time of verification is $O(n^k)$. We will show it is NP-hard by reducing vertex cover problem for cubic graphs to it, which is known to be NP-complete \cite{cubicVC}. Applying the same reduction given in Theorem 1 of \cite{DBLP:journals/dam/BresarKKS11}, for the input instance of a cubic graph $G$ we get the reduced graph instance $G'$. Since the proof of reduction given in Theorem 2 of \cite{DBLP:journals/dam/BresarKKS11} is independent of the $\Delta(G)$ and $\Delta(G') = \Delta(G) + 1 = 4$, hence the reduction implies NP-hardness for k-path vertex cover on bounded degree graphs too.
\end{proof}
\begin{corollary}
Minimum weighted k-path vertex cover for bounded degree graphs is NP-hard.
\end{corollary}

\begin{lemma}\label{lem_1}
Unless $P=NP$, reoptimization under vertex insertion of $k$-$path$ vertex cover problem for bounded degree graphs $(\pi)$ does not admit a polynomial time optimal algorithm.
\end{lemma}

\begin{proof}
Let Oracle $A$ output optimal solution for $Reopt(\pi)$ under vertex insertion for bounded degree graphs. Given oracle access to A, we design a polynomial time algorithm $B$ to obtain an optimal solution for problem $\pi$ on an arbitrary non-empty graph $G = (V,E)$. Here $|V| = n$.
 
\begin{algorithm}
\caption{$B(\pi,G)$}\label{alg:ReductB}
\begin{algorithmic}[1]
\State Let $\langle v_1,\dots v_n\rangle$ be an arbitrary sequence of $V(G)$
\State $G_1$ = $(V_1,E_1)$ = $(\{v_1\},\phi)$
\State $ALG(G_1) = \phi$
\For{ $i$ \textbf{from} $2$ \textbf{to} $n$ }
\State $G_i = \{V_i,E_i\}$
\State $E^a_i = \{ (u,v_i) \mid u \in V(G_{i-1})$ and $ (u,v_i) \in E(G)\}$
\State $G_i = (V_{i-1} \cup \{v_i\},E_{i-1} \cup E^a_i)$
\State $ALG(G_i) = A(G_{i-1},G_i,ALG(G_{i-1})) $
\EndFor\label{reoptendwhile}
\State \textbf{return} $ALG(G_{n})$
\end{algorithmic}
\end{algorithm}
We will claim by the principle of mathematical induction that $ALG(G_i)$ is the optimal solution for $G_i$, where $1 \leq i \leq n$. For the base case of $G_1=(\{v_1\},\phi)$, $\phi$ is the optimum for $G_1$. For the inductive step, assume $ALG(G_{i-1})$ is optimum for $G_{i-1}$. Since $G_i = (V_{i-1} \cup \{v_i\},E_{i-1} \cup E^a_i)$ for $v_i \notin V[G_{i-1}]$ and $E^a_i \subseteq (V[G_{i-1}] \times \{v_i\})$, the change is a valid vertex insertion. Therefore, algorithm $A(G_{i-1},G_i,ALG(G_{i-1})$ indeed outputs the optimal solution for $G_i$ as $ALG(G_i)$, thus proving the induction hypothesis. Next, we claim that $G_n = G$. Clearly $V(G_n) = V(G)$, thus we only need to show $E(G_n) = E(G)$. $\forall$ $e\in E(G_n)$, $E^a_i$ is the only set which contributes to $E(G_n)$, enforcing $e \in E(G)$. For any $(v_i,v_j) \in E(G)$ such that $v_i$ comes before $v_j$ in the vertex sequence, in the $j^{th}$ iteration of loop $(v_i,v_j) \in E^a_j$, thus $(v_i,v_j) \in E(G_n)$. Thus the algorithm $B$ outputs $OPT(G)$. All steps in the algorithm $B$ are polynomial in input size. Hence, we obtain a polynomial time Turing reduction from $\pi$ to $Reopt(\pi)$ under vertex insertion. Thus, the proof of $Lemma$ follows since $\pi$ is known to be an NP-complete problem (refer Lemma 3 in Appendix A).
\end{proof}
\section{Approximation algorithms for weighted $k$-$PVCP$}
In this section, we give the approximation algorithms for weighted $k$-$PVCP$ $(\pi_k)$ on input graph $G=(V,E)$. Here, $|V| = n$ and subroutine $getKPath(G)$ outputs a path of order k in $G$ if it is present and $\phi$ otherwise. 

\begin{algorithm}
\caption{$Approx$-$Alg(\pi_k,G)$}\label{alg:Reoptw4}
\begin{algorithmic}[1]
\State $V' = V(G)$
\State $ALG(G)=\phi$
\State $P = getKPath(G[V'])$
\While{$(P \neq \phi)$}
\State $v_m =$ minimum weight vertex in $P$
\State $ALG(G) = ALG(G) \cup  \{v_m\}$
\State $V' = V' - \{v_m\}$
\State $P = getKPath(G(V'))$
\EndWhile
\State \textbf{return} $ALG(G)$
\end{algorithmic}
\end{algorithm}

\begin{theorem}
Algorithm $Approx$-$Alg$ achieves $n$-approximation for $\pi_k$ on graph $G=(V,E)$ and runs in $O(2^k n^{O(1)})$ time, where $n$ is $|V|$.
\end{theorem}

\begin{proof}
In the $i^{th}$ iteration of while loop, we find a $k$-$path$, $P_i$, by calling subroutine $getKPath$. Let the vertex with minimum weight in $P_i$ be $v_m^i$. Since, $OPT(G)$ must include at least one vertex from $P_i$, implying $w(v_{m}^{i}) \leq w(OPT(G))$, where $OPT(G)$ is an optimal solution for $\pi_k$. 
Algorithm $Approx$-$Alg$ always terminates because $|ALG(G)|$ is at most $n-k+1$. The while loop continues to run till $ALG(G)$ is a feasible $k$-$path$ vertex cover.
$$w(ALG(G)) = \sum_{n=1}^{|ALG(G)|} w(v_{m}^{i}) \leq \sum_{n=1}^{|ALG(G)|} w(OPT(G))  \leq (n-k+1)w(OPT(G))$$
Thus, Algorithm $Approx$-$Alg$ achieves $n$-approximation for the problem $\pi_k$ on $G$. The function $getKPath(G)$ uses color coding algorithm for finding a $path$ on $k$ vertices \cite{DBLP:conf/stoc/AlonYZ94} and runs in $O(2^k n^{O(1)})$. Hence, algorithm $Approx$-$Alg$ runs in time $O(2^k n^{O(1)})$.
\end{proof}

\begin{theorem}
There is a $k$-approximation algorithm for weighted $k$-$PVCP$ $(\pi_k)$ for graph $G=(V,E)$.
\end{theorem}
\begin{proof}
Let $P_k$ be the collection of all k-paths in a graph $G$. Note that $|P_k|$ is at most $n^k$. For each vertex $v \in V$, we denote subset $v_P$ to be the collection of $k$-$paths$ containing $v$. $\pi_k$ is a special case of set cover problem where $P_k$ is the universe of elements and $ V_P = \{v_P| v \in V\}$ is family of subsets. Since each path contains exactly $k$ vertices, the frequency of the element is $k$ in the family of subsets. Thus, we have $k$-approximation algorithm from Theorem 15.3 in \cite{DBLP:books/daglib/0004338}.  
\end{proof}

\end{document}